\documentclass[11pt]{article}
\usepackage[margin=0.86in]{geometry}
\usepackage[T1]{fontenc}
\usepackage{lmodern}
\usepackage{microtype}
\usepackage{amsmath,amssymb,amsthm,mathtools}
\usepackage{booktabs,array,tabularx,longtable}
\usepackage{enumitem}
\usepackage{xcolor}
\usepackage{xurl}
\usepackage{tikz}
\usetikzlibrary{calc,positioning,arrows.meta}
\usepackage[hidelinks]{hyperref}
\usepackage[nameinlink,noabbrev]{cleveref}

\definecolor{loopblue}{RGB}{24,90,190}
\definecolor{iceblue}{RGB}{205,235,248}
\definecolor{witnessred}{RGB}{190,35,45}
\definecolor{gridgray}{RGB}{185,185,185}

\newtheorem{theorem}{Theorem}[section]
\newtheorem{lemma}[theorem]{Lemma}
\newtheorem{corollary}[theorem]{Corollary}

\theoremstyle{definition}

\theoremstyle{remark}

\crefname{theorem}{theorem}{theorems}
\Crefname{theorem}{Theorem}{Theorems}
\crefname{lemma}{lemma}{lemmas}
\Crefname{lemma}{Lemma}{Lemmas}
\crefname{corollary}{corollary}{corollaries}
\Crefname{corollary}{Corollary}{Corollaries}
\crefname{proposition}{proposition}{propositions}
\Crefname{proposition}{Proposition}{Propositions}
\crefname{definition}{definition}{definitions}
\Crefname{definition}{Definition}{Definitions}
\crefname{remark}{remark}{remarks}
\Crefname{remark}{Remark}{Remarks}

\newcommand{\FNP}{\mathrm{FNP}}

\title{Ice Walk is ASP-Complete}
\author{Papangkorn Apinyanon\\
Massachusetts Institute of Technology\\
\texttt{raya@mit.edu}}
\date{}

\begin{document}
\maketitle

\noindent
    \begin{abstract}
        We prove that the solution-search problem for the pencil puzzle Ice Walk
        is ASP-complete.  Our reduction maps Hamiltonian cycles in an undirected
        maximum-degree-3 spanning subgraph of a rectangular grid graph bijectively
        to Ice Walk solutions.
    \end{abstract}

    \section{Introduction}

    \subsection{Ice Walk}
    Ice Walk is a loop puzzle introduced by Broken Sign Games
    \cite{IceWalkRules}.  An instance is a rectangular grid whose cells are
    either icy or dry; some dry cells carry a positive-integer clue.  A solution is
    a single loop through centres of orthogonally adjacent cells that passes
    through every clue.  The loop may not turn on ice.  Two perpendicular
    segments may intersect only in an icy cell, and the loop may neither turn
    at such an intersection nor overlap itself.  A clue $q$ specifies that the
    maximal consecutive dry section of the loop containing that clue has
    $q$ cells.  Two solutions are
    considered identical when they use the same unoriented set of unit grid
    edges; the starting point and direction of traversal are ignored.

    \subsection{ASP-Completeness}
    We view Ice Walk as a search problem whose solutions are the completed
    loops.  An \emph{ASP reduction} between search problems is a polynomial-time
    map on instances together with polynomial-time, mutually inverse maps
    between the solution sets of each input instance and its image.  A problem
    is \emph{ASP-complete} if it is in $\FNP$ and every NP search problem has
    such a reduction to it.  Here $\FNP$ is the class of search problems for
    which every proposed solution has polynomial length and can be checked in
    polynomial time.  In particular, ASP-completeness records more than
    NP-hardness: it preserves each individual solution rather than merely the
    existence of one \cite{YatoSeta2003,MITHardness2026}.

    \subsection{Related Work}
    Yato and Seta introduced ASP-completeness as a framework for ``another
    solution'' problems \cite{YatoSeta2003}.  Tang's T-metacell framework gives
    a convenient route from grid-graph Hamiltonicity to many loop-puzzle
    hardness results \cite{Tang2022}.  More recently, MIT Hardness Group proved
    ASP-completeness for Hamiltonicity in several grid-graph classes, including
    the restricted source problem used here \cite{MITHardness2026}.  Our
    reduction is direct: instead of simulating a source vertex by a metacell,
    it realizes source edges as straight visibility tracks between clues.

    \subsection{The grid-graph source problem}
    Let $R_{r,c}$ be the full $r\times c$ rectangular grid graph.  A
    \emph{maximum-degree-3 spanning subgraph of $R_{r,c}$} retains all $rc$ vertices,
    deletes an arbitrary subset of grid edges, and has maximum degree at most 3.

    \begin{theorem}[MIT Hardness Group]\label{thm:source}
        The solution-search problem for undirected Hamiltonian cycle in a
        maximum-degree-3 spanning subgraph of a rectangular grid graph is ASP-complete.
        It remains ASP-complete under the additional promise that every degree-3 vertex
        has a forced side edge \cite[Theorem~17]{MITHardness2026}.
    \end{theorem}

    \section{Reduction}
    \subsection{Construction}
    We represent $G$ directly.  Each source vertex becomes a dry 1-clue.
    Vertices in one maximal horizontal source component are placed
    on a common horizontal track, and vertices in one maximal vertical component
    on a common vertical track; different components receive different tracks.
    Every cell other than these clues is ice.  Thus a segment between two clues
    must travel straight along one track; the lemmas below show that precisely
    the source edges arise in this way.

    Let $G=(V,E)$ be a maximum-degree-3 spanning subgraph of a rectangular grid
    graph with $r$ rows and $c$ columns. 
    Each vertex in $G$ is referred to by its zero-based indices: 
    \[
        V=\{(i,j):0\le i<r,\ 0\le j<c\}.
    \]

    For each row $i$, partition its vertices into maximal intervals connected by
    present horizontal edges.  Let $h_i(j)$ be the left-to-right index, starting at
    0, of the interval containing $(i,j)$.
    Similarly, in each column $j$, partition the vertices into maximal intervals
    connected by present vertical edges, and let $v_j(i)$ be the top-to-bottom
    zero-based index of the interval containing $(i,j)$.

    The construction converts $G$ into an Ice Walk instance of width
    $c(2r + 1)$ and height $r(2c + 1)$.

    For every vertex $(i, j)$ in $G$, a dry cell with number $1$ sits at
    \[
        (X(i,j),Y(i,j))
        =\bigl(j(2r+1)+2v_j(i)+1,\; i(2c+1)+2h_i(j)+1\bigr),
    \]
    where all indices are zero-based.  All other cells are ice cells.

    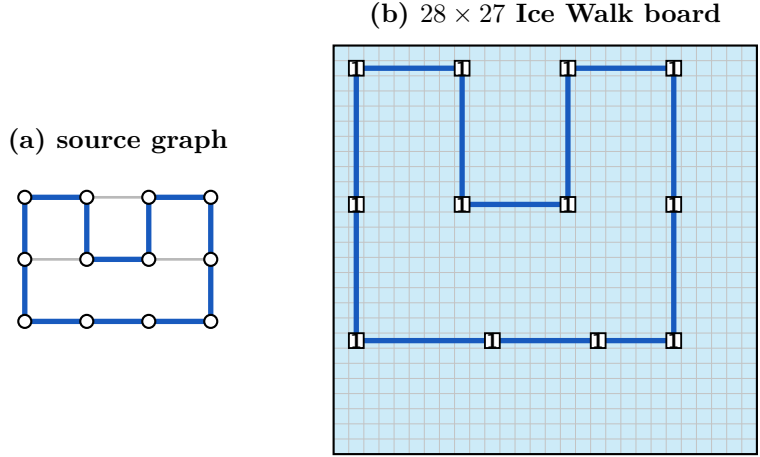
\begin{figure}[t]
        \centering
        \begin{tikzpicture}[x=0.82cm,y=-0.82cm,baseline=(current bounding box.center)]
            \node[font=\small\bfseries] at (1.5,-0.9) {(a) source graph};
            \draw[gridgray,line width=1.0pt] (1,0)--(2,0);
            \draw[gridgray,line width=1.0pt] (0,1)--(1,1);
            \draw[gridgray,line width=1.0pt] (2,1)--(3,1);
            \draw[loopblue,line width=2.0pt,line cap=round,line join=round]
                (0,0)--(1,0)--(1,1)--(2,1)--(2,0)--(3,0)--(3,1)--
                (3,2)--(2,2)--(1,2)--(0,2)--(0,1)--cycle;
            \foreach \x/\y in {0/0,1/0,2/0,3/0,0/1,1/1,2/1,3/1,0/2,1/2,2/2,3/2}{
                \fill[white] (\x,\y) circle[radius=0.105];
                \draw[black,line width=0.8pt] (\x,\y) circle[radius=0.105];
            }
        \end{tikzpicture}
        \hspace{2.6em}
        \begin{tikzpicture}[x=0.20cm,y=-0.20cm,baseline=(current bounding box.center)]
            \node[font=\small\bfseries] at (13.5,-2.5) {(b) $28\times27$ Ice Walk board};
            \fill[iceblue] (-0.5,-0.5) rectangle (27.5,26.5);
            \foreach \x in {0,...,27}
                \draw[gridgray!85,line width=0.25pt] ({\x+0.5},-0.5) -- ({\x+0.5},26.5);
            \foreach \y in {0,...,26}
                \draw[gridgray!85,line width=0.25pt] (-0.5,{\y+0.5}) -- (27.5,{\y+0.5});
            \draw[black,line width=0.8pt] (-0.5,-0.5) rectangle (27.5,26.5);
            \draw[loopblue,line width=2.0pt,line cap=round,line join=round]
                (1,1)--(8,1)--(8,10)--(15,10)--(15,1)--(22,1)--
                (22,10)--(22,19)--(17,19)--(10,19)--(1,19)--(1,10)--cycle;
            \foreach \x/\y in {1/1,8/1,15/1,22/1,1/10,8/10,15/10,22/10,1/19,10/19,17/19,22/19}{
                \fill[white] (\x-0.49,\y-0.49) rectangle (\x+0.49,\y+0.49);
                \draw[black,line width=0.7pt] (\x-0.49,\y-0.49) rectangle (\x+0.49,\y+0.49);
                \node[font=\sffamily\bfseries\scriptsize,text=black] at (\x,\y) {1};
            }
        \end{tikzpicture}
        \caption{The two absent vertical source edges
            between the second and third rows put their endpoints on different
            vertical tracks.  The Hamiltonian cycle maps to the 
            Ice Walk loop.}
        \label{fig:ice-example}
    \end{figure}

    \subsection{Coordinate and visibility lemmas}

    \begin{lemma}\label[lemma]{lem:separation}
        The relative row order and column order are preserved through the conversion.
    \end{lemma}
    \begin{proof}
        Consider any pair of two distinct vertices $(i_1, j_1)$ and
        $(i_2, j_2)$ in $G$.

        If their ordinates differ, without loss of generality let $i_1 < i_2$,
        then
        \[
            Y(i_2,j_2)-Y(i_1,j_1)
            =(i_2-i_1)(2c+1)+2\bigl(h_{i_2}(j_2)-h_{i_1}(j_1)\bigr)
            \ge (2c+1)-2(c-1)=3.
        \]
        Thus their ordinates also differ in the converted Ice Walk instance.

        Similarly, if their abscissas in $G$ differ, without loss of generality
        let $j_1 < j_2$.  Then
        \[
            X(i_2,j_2)-X(i_1,j_1)
            =(j_2-j_1)(2r+1)+2\bigl(v_{j_2}(i_2)-v_{j_1}(i_1)\bigr)
            \ge (2r+1)-2(r-1)=3.
        \]
        Thus their abscissas also differ in the converted Ice Walk instance.
    \end{proof}

    \begin{lemma}\label[lemma]{lem:isolation}
        A cell adjacent to a dry cell is an ice cell.
    \end{lemma}
    \begin{proof}
        Two distinct source vertices differ in their row or their column.
        By \cref{lem:separation}, the corresponding dry cells differ by at least
        $3$ in at least one coordinate. Hence they cannot be orthogonally adjacent.
    \end{proof}

    \begin{lemma}\label[lemma]{lem:row-track}
        Two dry cells are in the same row iff source vertices are in the
        same row and horizontal component.
    \end{lemma}
    \begin{proof}
        The formula for $Y(i,j)$ depends only on $i$ and $h_i(j)$, proving
        sufficiency.  Conversely, if two converted dry cells have the same
        ordinate, \cref{lem:separation} first implies that their source vertices
        have the same row.  With the row fixed, equality of their ordinates
        implies equality of their horizontal-component indices.
    \end{proof}

    \begin{lemma}\label[lemma]{lem:column-track}
        Two dry cells are in the same column iff source vertices are in the
        same column and vertical component.
    \end{lemma}
    \begin{proof}
        The formula for $X(i,j)$ depends only on $j$ and $v_j(i)$, proving
        sufficiency.  Conversely, if two converted dry cells have the same
        abscissa, \cref{lem:separation} implies that their source vertices have
        the same column.  Equality of their abscissas then implies equality of
        their vertical-component indices.
    \end{proof}

    \begin{lemma}[Track-edge correspondence]\label[lemma]{lem:visibility}
        Two dry cells are consecutive in a common row or column iff
        their corresponding source vertices are adjacent in $G$.
    \end{lemma}

    \begin{proof}
        Consider the horizontal case first.

        Let two consecutive dry cells lie in one row.  By \cref{lem:row-track},
        their source vertices lie in one horizontal component.  If their source
        columns were not consecutive, the intervening source vertex would also
        lie in that component, and its dry cell would lie strictly between the
        two, a contradiction.  Thus the source vertices are adjacent in $G$.

        Conversely, suppose two source vertices are joined by a horizontal edge.
        They lie in one horizontal component, so their dry cells have the same
        ordinate.  A dry cell strictly between them would, by
        \cref{lem:row-track}, come from the same source row; its source column
        would then lie strictly between two consecutive columns, which is
        impossible.  Hence the two dry cells are consecutive.

        The vertical case is identical, using \cref{lem:column-track}.
    \end{proof}

    \begin{lemma}[Bijection between cycles and solutions]\label[lemma]{lem:bijection}
        A Hamiltonian cycle in graph $G$ corresponds to a unique solution of
        the converted Ice Walk instance.
    \end{lemma}

    \begin{proof}
        Given a Hamiltonian cycle of $G$, replace every selected source edge by
        the straight track between the corresponding dry cells.  By
        \cref{lem:visibility}, every such track joins consecutive dry cells.
        At each dry cell exactly two selected tracks meet, so these tracks form
        one loop.  Distinct collinear tracks have disjoint interiors, and a
        perpendicular interior intersection can occur only on ice: otherwise
        the dry cell at the intersection would contradict consecutiveness.
        Thus the trace obeys the turning, crossing, and overlap rules of Ice
        Walk.  By \cref{lem:isolation}, every non-icy section contains exactly
        one dry cell, so each 1-clue is satisfied.

        Conversely, cut an Ice Walk solution at every dry cell.  Each resulting
        subpath has dry endpoints and only ice as intermediate cells.  It cannot
        turn, and hence is a straight track.  By \cref{lem:visibility}, its two
        endpoints correspond to an edge of $G$.  The tracks incident with each
        dry cell give degree two at the corresponding source vertex, and the
        single puzzle loop makes the selected edges connected.  They therefore
        form a connected spanning 2-regular subgraph of $G$, namely a Hamiltonian
        cycle.

        Both conversions simply replace edges by their tracks, or tracks by
        their endpoint edges, so they are polynomial-time inverses.  Hence the
        correspondence is a bijection.
    \end{proof}

    \begin{theorem}\label{thm:ice}
        The Ice Walk solution-search problem is ASP-complete.
    \end{theorem}
    \begin{proof}
        Each graph $G$ that is a maximum-degree-3 spanning subgraph of a
        rectangular grid graph with $r$ rows and $c$ columns can be converted
        into an Ice Walk instance with $r(2c+1)$ rows and $c(2r+1)$ columns.

        Hence the Ice Walk board has area $rc(2c+1)(2r+1)$, polynomial in the
        source size.

        By \cref{thm:source} and \cref{lem:bijection}, this construction is an
        ASP reduction to Ice Walk.  A solution can be encoded by its used unit
        grid edges.  Non-overlap bounds this encoding by the board area, and a
        verifier can in polynomial time check the clue lengths, the turning and
        crossing rules, and that the resulting trace is one loop.  Thus Ice
        Walk belongs to $\FNP$, and is ASP-complete.
    \end{proof}

    \begin{corollary}\label[corollary]{cor:ice-restricted}
        Ice Walk remains ASP-complete even when every non-ice cell is numbered.
    \end{corollary}
    \begin{proof}
        The instances constructed by the reduction have this form.
    \end{proof}

\end{document}